\documentclass[a4paper,USenglish,cleveref, autoref, thm-restate]{lipics-v2021}
\usepackage{amsmath}
\usepackage{amssymb}
\usepackage{amsfonts}
\usepackage{mathrsfs}
\usepackage{caption}
\usepackage{subcaption}
\usepackage[T1]{fontenc}
\usepackage{graphicx}
\usepackage{color}
\usepackage{hyperref}
\usepackage{bbding}

\usepackage{tikz}
\usetikzlibrary{automata,positioning, arrows.meta, shapes.geometric,
	arrows,decorations.pathreplacing,calligraphy}
\tikzset{
	->,  
	>=stealth, 
	shorten >=1pt, 
	node distance=2cm, 
	every state/.style={thick, fill=gray!10}, 
	initial text=$ $, 
	initial distance = .5cm,
	accepting distance = .5cm,
	accepting text=$ $, 
}

\newcommand{\A}{\mathcal{A}}
\newcommand{\B}{\mathcal{B}}
\newcommand{\N}{\mathbb{N}}
\newcommand{\cmin}{c_\mathit{min}}
\newcommand{\cmax}{c_\mathit{max}}
\newcommand{\cons}{\mathscr{C}}
\newcommand{\jmin}{\widehat{J}_\mathit{min}}
\newcommand{\jmax}{\widehat{J}_\mathit{max}}
\newcommand{\alphamax}{\widehat{\alpha}_\mathit{max}}
\newcommand{\cof}{\mathit{cf}}

\title{Constructible Words Characterize Rational Languages of Words Indexed by Scattered Linear Orderings }
\titlerunning{Constructible Words and Automata on Scattered Linear Orderings}

\author{Thomas Braipson}{University of Liège, Montefiore Institute B28 Allée de la découverte 9 4000 Liège, Belgium}{thomas.braipson@uliege.be}{0009-0003-0543-5265}{has been granted with a FRIA grant.}
\author{Tom Clara}{University of Liège, Montefiore Institute B28 Allée de la découverte 9 4000 Liège, Belgium}{tom.clara@uliege.be}{0009-0009-2938-1410}{is a Research Fellow of the F.R.S.-FNRS.}
\authorrunning{Th.~Braipson and T.~Clara}
\Copyright{Thomas Braipson and Tom Clara}



\acknowledgements{The authors wish to thank Olivier Carton for discussing with them and advertising Colcombet's theorem. They also wish to thank the anonymous reviewers for their comments and ideas to go beyond this paper.}

\funding{This work is partially supported by the FNRS-DFG PDR Weave (SMT-ART) grant 40019202.}

\EventEditors{Michal Kouck\'{y} and Daniela Petri\c{s}an}
\EventNoEds{2}
\EventLongTitle{51st International Symposium on Mathematical Foundations of Computer Science (MFCS 2026)}
\EventShortTitle{MFCS 2026}
\EventAcronym{MFCS}
\EventYear{2026}
\EventDate{August 24--28, 2026}
\EventLocation{Paris, France}
\EventLogo{}
\SeriesVolume{386}
\ArticleNo{20}

\pdfoutput=1 
\hideLIPIcs
\nolinenumbers
\begin{document}
	
\maketitle

\begin{abstract}
	Automata on linear orderings are finite-state automata introduced by Bruyère and Carton as a broad generalization of finite, infinite and transfinite-word automata. In this context, a word is defined as a function from a linear ordering to a finite alphabet. This general definition can make automata on linear orderings difficult to reason about. In this work, we introduce constructible words as an intuitive way of tackling this difficulty. These words can be obtained by a finite number of applications of simple operators and thus admit a finite notation. We show that a rational language of words indexed by scattered (countable and uncountable) linear orderings is characterized by its constructible words. Our proof of this result relies on an interesting theorem of semigroup theory due to Colcombet. We expect this property to be useful in future theoretical developments about automata on scattered linear orderings.
\end{abstract}

\section{Introduction}
Automata on linear orderings~\cite{BC07,BC06,epsilon-automata} were introduced by Bruyère and Carton and consist of a broad and elegant generalization of finite-word automata, infinite-word automata~\cite{WolfgangThomas}, automata on bi-infinite words~\cite{PN82} and transfinite words~\cite{Woj85,Chou78}. In the context of automata on linear orderings, a word is defined as a function from any linear (i.e., total) ordering to a finite alphabet. This flexible definition generalizes the previously mentioned common notions of words.

Reasoning about automata on linear orderings (for instance, showing that two automata accept exactly the same words) can be tedious, since some complex words do not admit a natural mental representation. In this work, we introduce a theoretical tool that addresses this problem. We first define a \emph{constructible} word as a word that admits a canonical, finite representation. We then prove that automata on scattered (countable and uncountable) linear orderings that accept the same constructible words are equivalent. Stated differently, rational languages of words indexed by scattered linear orderings are characterized by their constructible words. We expect this result to be useful in further developments.

Our work extends the characterization of $\omega$-regular languages by ultimately periodic words~\cite[Fact 1]{ultimately-periodic-words}. A simple proof of this classical result relies on the closure of $\omega$-regular languages under complementation~\cite{Buc62}. However, rational languages of words indexed by scattered (countable and uncountable) linear orderings are not closed under complementation~\cite{Ris05}, hence our proof relies on different principles. 

Sketchily, our proof is based on the fact that some words cannot be distinguished by a finite-state automaton, which naturally induces an equivalence relation on words. This idea, already used by Büchi in~\cite{Buc62}, was then adapted by Carton and Rispal to prove that rational languages on finite-rank scattered linear orderings are closed under complementation~\cite{RispalFiniteRank}. We reuse this idea in the case of words on scattered linear orderings, and we show that each equivalence class induced by an automaton contains a constructible word.

Much work has already been carried out regarding rational languages on \emph{countable} scattered linear orderings. These languages were shown to coincide with those recognizable by a special form of semigroup~\cite[Theorem 14]{Ris05}. As a consequence, they are closed under complementation~\cite[Theorem 16]{Ris05}, and definable in the monadic second-order theory of order (MSO)~\cite{automata-MSO}. Carton also showed that any non-empty rational language on scattered linear orderings contains a countable constructible word~\cite[Theorem 2]{carton-emptiness}. In the case of countable scattered words, our result follows immediately from these previous works. It can also be deduced from several results about the decidability~\cite{Shelah_1975} and expressivity~\cite{CCP} of MSO on countable domains. Hence our main contribution is in fact to provide a characterization of rational languages that goes beyond countability. 

Our proofs exploit connections between rational languages and semigroups. It turns out that the algebraic properties of rational languages on unrestricted (i.e., countable and uncountable) scattered linear orderings are not as strong as they are in the countable case. Nevertheless, they remain useful, especially when combined with powerful semigroup results such as the factorization forests theorem due to Colcombet~\cite[Theorem 4]{Colcombet-ramseyan-splits}, which is a generalization of Simon's factorization forests theorem~\cite{simon}. Our work strongly relies on this result, and provides an example of use of its infinite version.

The remainder of this article is organized as follows. Section~\ref{sec:basic} recalls some notions related to automata on linear orderings. Then, in Section~\ref{sec:constructible-words} we formally define the notion of \emph{constructible} word. Section~\ref{sec:equivalence-relation} introduces an equivalence relation on words and then motivates the link between automata and semigroups. Finally, in Section~\ref{sec:constructible-words-characterization}, we present Colcombet's theorem and apply it to show that two automata on scattered linear orderings are equivalent if and only if they accept the same constructible words. 

\section{Preliminary Notions}\label{sec:basic}

\subsection{Orderings and Cuts}
In this section, we recall some definitions on linear orderings. Further details can be found in~\cite{Ros82}. A \emph{linear ordering} is a set $J$ equipped with a binary relation $<_J$ that is
\begin{itemize}
	\item total (for every $j \neq k \in J$, either $j <_J k$ or $k <_J j$),
	\item irreflexive (for every $j \in J$, $j \not<_J j$),
	\item antisymmetric (for every $j, k\in J$, $j <_J k$ implies $k \not<_J j$),
	\item and transitive (for every $j,k, \ell \in J$, $j <_J k$ and $k <_J \ell$ imply $j <_J \ell$).
	\end{itemize}
We simply write $<$ in place of $<_J$ when the ordering $J$ is clear from the context.

Given a linear ordering $J$, a subordering of $J$ is a subset $I\subseteq J$ equipped with the order relation $<_I$ defined as the restriction of $<_J$ to $I$. A linear ordering $J$ is \emph{dense} if for any $j < k\in J$, there exists $\ell\in J$ such that $j < \ell < k$. A linear ordering is \emph{scattered} if it does not contain any infinite dense subordering\footnote{Notice that according to the definition of density, a singleton is a dense linear ordering. Hence there exist dense linear orderings that are not infinite.}. This work focuses on scattered (countable and uncountable) linear orderings.

Two elements $j < k \in J$ are said to be
\emph{consecutive} if there does not exist $\ell \in J$ such that $j <
\ell < k$. In such a case, $k$ is the \emph{successor} of $j$, and $j$
the \emph{predecessor} of $k$.

A \emph{cut} of a linear ordering $J$ is an ordered pair $(K,L)$ such that $\{K,L\}$ is a partition of $J$, and for every $k \in K$ and $\ell \in L$, one has $k < \ell$. The set of cuts of $J$ is
denoted by $\widehat{J}$. This set is itself a linear ordering under the relation $<_{\widehat{J}}$ such that $(K_1, L_1)
<_{\widehat{J}} (K_2, L_2)$ if and only if $K_1 \subsetneq K_2$ (or equivalently
$L_2 \subsetneq L_1$). The linear ordering $\widehat{J}$ has a least and a greatest element. These are respectively called the \emph{first cut}
$\widehat{J}_{\mathit{min}} = (\emptyset, J)$ and the \emph{last cut}
$\widehat{J}_{\mathit{max}} = (J, \emptyset)$. The set $\widehat{J}
\setminus \{ \widehat{J}_{\mathit{min}}, \widehat{J}_{\mathit{max}}
\}$ of \emph{non-extremal} cuts of $J$ is denoted by $\widehat{J}^*$.  

Given two linear orderings $J_1$ and $J_2$, the ordering $J = J_1 +
J_2$ is defined as the ordering obtained by placing the elements of
$J_1$ before those of $J_2$, both sets keeping their internal order
among their elements. Formally, $J$ can be defined as the set $(J_1
\times \{ 1 \}) \cup (J_2 \times \{ 2 \})$ such that $(j, m) < (k, n)$
if and only if either $m <_{\mathbb{N}} n$, or $m = n$ and $j <_{J_m} k$. 
More generally, given the linear orderings $J$ and $K_j$ for each $j \in J$, the
ordering $\sum_{j \in J} K_j$ is obtained, schematically, by
replacing every element $j$ of $J$ by the corresponding set $K_j$,
keeping the internal order among the elements of this set. Formally,
we have $\sum_{j \in J} K_j = \{ (k, j) \mid j \in J \,\wedge\, k \in
K_j \}$, with $(k_1, j_1) < (k_2, j_2)$ if and only if either $j_1 <_J j_2$, or
$j_1 = j_2$ and $k_1 <_{K_{j_1}} k_2$. 

Given a linear ordering $J$, the \emph{reverse} linear ordering $-J$ is obtained by keeping the same elements as $J$ and by flipping the order relation $<_J$. For two elements $i,j\in -J$, we thus have $i<_{-J}j$ if and only if $j<_{J}i$.

\begin{example}
    The linear ordering $\mathbb{N} + \mathbb{Z} + 2$ and its cuts are depicted in Figure~\ref{fig:example-LO}, where dots represent the elements of $\mathbb{N}+\mathbb{Z}+2$ and vertical bars represent its cuts.
\end{example}
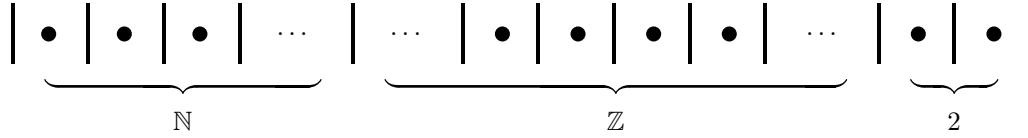
\begin{figure}[h!]
    \centering
    \begin{tikzpicture}[on grid, auto]
        \foreach \i in {0,1,...,3} {
			\fill (\i-0.5,-.4) rectangle (\i-0.45, 0.4);}
        \foreach \i in {0,1,2} {
        \node[circle, fill=black, inner sep=2pt] at (\i,0) {};}
        \node at (3.25,0) {$\dots$};
        \draw [-,decorate,decoration={calligraphic brace,amplitude=5pt,mirror,raise=4ex}, line width = 1pt]
      (-.45+.4,0) -- (4.-.4,0) node[midway,yshift=-4em]{$\mathbb{N}$};
      \fill (1.25,-4.7ex+1pt) rectangle (1.6,-4.7ex);
      \fill (3,-4.7ex+1pt) rectangle (3.45,-4.7ex);
            
        \fill (4.5-0.5,-.4) rectangle (4.5-0.45, 0.4);
        \node at (4.75,0) {$\dots$};
        \foreach \i in {0,1,...,4} {
			\fill (\i+5.5,-.4) rectangle (\i+5.45, 0.4);}
        \foreach \i in {1,2,...,4} {
        \node[circle, fill=black, inner sep=2pt] at (\i+5,0) {};}
        \node at (10.25,0) {$\dots$};
        \draw [-,decorate,decoration={calligraphic brace,amplitude=5pt,mirror,raise=4ex}, line width = 1pt]
      (4.05+.4,0) -- (11-.05-.4,0) node[midway,yshift=-4em]{$\mathbb{Z}$};
      \fill (6.5,-4.7ex+1pt) rectangle (7.35,-4.7ex);
      \fill (10,-4.7ex+1pt) rectangle (10.4,-4.7ex);

        \fill (11,-.4) rectangle (11-.05, 0.4);
        \foreach \i in {12,13} {
			\fill (\i,-.4) rectangle (\i-.05, 0.4);}
        \foreach \i in {11,12} {
        \node[circle, fill=black, inner sep=2pt] at (\i+.5,0) {};}
        \draw [-,decorate,decoration={ calligraphic brace,amplitude=5pt,mirror,raise=4ex}, line width = 1 pt]
      (11+.4,0) -- (13-.05-.4,0) node[midway,yshift=-4em]{$2$};
      \fill (11.5+.2,-4.7ex+1pt) rectangle (11.8,-4.7ex);
      \fill (12.7-.4,-4.7ex+1pt) rectangle (12.8-.4,-4.7ex);
    \end{tikzpicture}
    \caption{The linear ordering $ \mathbb{N} + \mathbb{Z} + 2$ and its cuts.}
    \label{fig:example-LO}
\end{figure}

\subsection{Words Indexed by Linear Orderings}

This section introduces a few definitions taken from \cite{BC07}. A
\emph{word indexed by a linear ordering} is a function $w: J
\rightarrow \Sigma$ from a linear ordering $J$ to a finite alphabet
$\Sigma$. We say that $J$ is the \emph{length} of $w$, which is denoted by $J=|w|$, and that $\widehat{J}$ are its cuts. In the particular case where $J$ is finite, the word $w$ is simply a finite word. The notions of infinite and bi-infinite words correspond to the cases where $J=\N$ and $J=\mathbb{Z}$, respectively.
The empty word is the only word indexed by the empty linear ordering. In the remainder of this article, we will consider words indexed by arbitrary scattered (countable and uncountable) linear orderings, that we will call \emph{scattered words}. The class of all the scattered words on the alphabet $\Sigma$ is noted $\Sigma^{\textsf{scat}}$.

In the context of automata on linear orderings, two words $w_1$ and $w_2$ are said to be \emph{isomorphic} if there exists a bijection $\tau:|w_1|\to |w_2|$ that preserves both the symbols and the order: for all $j,k\in |w_1|$, we must have $w_1(j)=w_2(\tau(j))$ and $j<_{|w_1|}k\implies \tau(j)<_{|w_2|}\tau(k)$. In the remainder of this article, we do not distinguish isomorphic words, since the automata that we will introduce in Section~\ref{sec:automata} are not able to discriminate them. Note that the length of a word should therefore be an \emph{order type} rather than an ordering (see~\cite{BC07} for more details). In this work, we do not distinguish these two notions and only use the term ``linear ordering'', even though it is a slight abuse of language.

Given two words $w_1:K_1\to \Sigma$ and $w_2:K_2\to \Sigma$ (where $K_1$ and $K_2$ are assumed disjoint), the \emph{concatenation} of $w_1$ and $w_2$, denoted by $w_1\cdot w_2$ (or $w_1w_2$), is the word $w:K_1+K_2\to \Sigma$ defined as
\[w\bigl((k,n)\bigr)=\left\{\begin{array}{ll}
	w_1(k)\text{ if $n=1$}\\
	w_2(k)\text{ if $n=2$.}
\end{array}\right.\]
The concatenation can be considered as a particular case of a more general operation. Given a linear ordering $J$ and a set of words $\{w_j:K_j\to \Sigma\mid j\in J\}$, the concatenation of these words, denoted by
$\prod_{j\in J}w_j$,
is a word $w$ indexed by $\sum_{j\in J}K_j$, and such that $w\bigl((k, j)\bigr)=w_j(k)$.

Given a word $w$ indexed by $J$ and two cuts $c_1=(K_1,L_1)$ and $c_2=(K_2,L_2)$ of $J$ such that $c_1\leq_{\widehat{J}}c_2$, the word $w[c_1:c_2]$ is obtained by restricting $w$ to the linear ordering $L_1\cap K_2$. In the case where $c_1=c_2$, the word $w[c_1:c_2]$ is thus the empty word.

\begin{example}
    The word $w : \mathbb{N}+\mathbb{Z}+2 \to \{a,b\}$ that maps a single element of the part $\mathbb{Z}$ to the letter $b$ and the remaining elements to the letter $a$ is illustrated in Figure~\ref{fig:example-word}. Dots represent the elements of $\mathbb{N}+\mathbb{Z}+2$, vertical bars are its cuts, and the letters are the result of the mapping. Note that the value, in $\mathbb{Z}$, of the element mapped to $b$ is irrelevant since every possible choice leads to a word isomorphic to $w$.

\end{example}

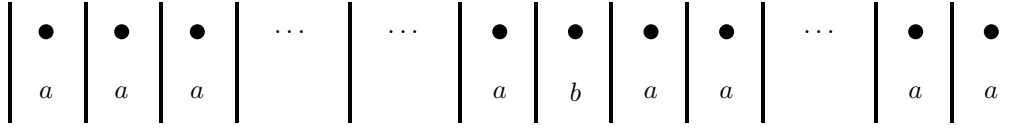
\begin{figure}[h!]
    \centering
    \begin{tikzpicture}[on grid, auto]
        \foreach \x in {0,2,6,8,11.5,1,9,12.5} {
        \node at (\x, -.8) {$a$};
        }
        \foreach \x in {7} {
        \node at (\x, -.8) {$b$};
        }
        \foreach \i in {0,1,...,3} {
			\fill (\i-0.5,-1.2) rectangle (\i-0.45, 0.4);}
        \foreach \i in {0,1,2} {
        \node[circle, fill=black, inner sep=2pt] at (\i,0) {};}
        \node at (3.25,0) {$\dots$};
            
        \fill (4.5-0.5,-1.2) rectangle (4.5-0.45, 0.4);
        \node at (4.75,0) {$\dots$};
        \foreach \i in {0,1,...,4} {
			\fill (\i+5.5,-1.2) rectangle (\i+5.45, 0.4);}
        \foreach \i in {1,2,...,4} {
        \node[circle, fill=black, inner sep=2pt] at (\i+5,0) {};}
        \node at (10.25,0) {$\dots$};

        \fill (11,-1.2) rectangle (11-.05, 0.4);
        \foreach \i in {12,13} {
			\fill (\i,-1.2) rectangle (\i-.05, 0.4);}
        \foreach \i in {11,12} {
        \node[circle, fill=black, inner sep=2pt] at (\i+.5,0) {};}
    \end{tikzpicture}
    \caption{Example of word indexed by the linear ordering $\mathbb{N}+\mathbb{Z}+2$.}
    \label{fig:example-word}
\end{figure}
\subsection{A Few Facts About Ordinals}

An \emph{ordinal} is a linear ordering $\alpha$ such that any non-empty subordering of $\alpha$ has a least element. In particular, this implies that each element $j\in \alpha$ that is not the greatest element has a successor in $\alpha$. The class of all the ordinals is linearly ordered. Informally speaking, given two ordinals $\alpha$ and $\beta$, we have $\alpha<\beta$ if $\alpha$ is isomorphic to a strict initial segment of $\beta$. That is, if there exists $b\in \beta$ such that an order-preserving bijection exists between $\alpha$ and $\{j\in \beta\mid j<_\beta b\}$. We refer to~\cite{sierpinski} for a thorough and formal introduction to ordinals and cardinals.

In this work, ordinals mainly appear as cofinalities of linear orderings and are denoted by Greek letters. Given a linear ordering $J$, a subordering $I\subseteq J$ is \emph{cofinal} in $J$ if
for all $j\in J$, there exists $i \in I$ such that $j\leq i$.
The \emph{cofinality} of $J$, denoted by $\cof(J)$, is defined as the least ordinal $\xi$ such that there exists a function $f:\xi\to J$ whose image is cofinal in $J$. The \emph{coinitiality} of $J$ is the cofinality of $-J$. 

An ordinal is called \emph{regular} if it equals the cofinality of some linear ordering. Alternatively, a regular ordinal can be defined as an ordinal equal to its cofinality. It follows from the definition of cofinality that a regular ordinal is necessarily the least ordinal of its cardinality, i.e., a regular ordinal is always \emph{initial}. The first regular ordinals are $0$ ($0$ is the cofinality of the empty linear ordering), $1$ ($1$ is the cofinality of any linear ordering that has a greatest element), $\omega$ (the first infinite ordinal), $\omega_1$ (the first uncountable ordinal, i.e., the first ordinal of cardinality $\aleph_1$) and $\omega_2$ (the first ordinal of cardinality $\aleph_2$).

Later, we will often need to consider concatenations of the form $\prod_{j\in \alpha} w$ or $\prod_{j\in -\alpha} w$ where $\alpha$ is an ordinal and $w$ is a word. These two concatenations are denoted concisely by $w^\alpha$ and $w^{-\alpha}$, respectively. Note that this definition leads to the usual notion of $\omega$-power. Indeed, in the particular case where $\alpha=\omega$, we get $\prod_{j\in \omega} w=w^\omega$.
\subsection{Automata on Linear Orderings}\label{sec:automata}
Automata on linear orderings, introduced in~\cite{BC07} and extended in~\cite{BC06,epsilon-automata}, are finite-state machines capable of reading words indexed by linear orderings. They can be seen as a broad generalization of infinite-word automata and transfinite-word automata.
\begin{definition}
	\label{def:automata-old}
	An \emph{automaton on linear orderings} is a tuple ${\cal A} = \left( Q,
	\Sigma, \Delta, I, F\right)$ where
	\begin{itemize}
		\item
		$Q$ is a finite set of \emph{states},
		\item
		$\Sigma$ is a finite \emph{alphabet},
		\item
		$\Delta \subseteq (Q \times \Sigma \times Q) \,\cup\, (Q \times 2^Q) \,\cup\,
		(2^Q \times Q)$ is a \emph{transition relation},
		\item
		$I \subseteq Q$ is a set of \emph{initial states}, and
		\item
		$F \subseteq Q$ is a set of \emph{final states}.
	\end{itemize}
\end{definition}

The transition relation $\Delta$ is divided into three parts. The transitions of the form $(q_1,a,q_2)$, where $q_1,q_2\in Q$ and $a\in \Sigma$, correspond to the usual transitions in finite-word automata. They are called \emph{successor transitions}. Those of the form $(q,P)$, with $q\in Q$ and $P\in 2^Q$, are called \emph{right-limit} transitions. We often denote them by $q\to P$, and we say that $P$ is the \emph{right-limit set} of the transition. Symmetrically, the transitions of the form $(P,q)$ are \emph{left-limit} transitions, denoted by $P\to q$, and $P$ is said to be the \emph{left-limit set} of the transition.

Let us now define the semantics of automata on linear orderings.
\begin{definition}
	\label{def:semantics-old}
	A \emph{run} of ${\cal A}$
	reading a word $w$ of length $J$ is a function $\rho: \widehat{J} \rightarrow Q$ that satisfies the following conditions:
	\begin{itemize}
		\item
		$\rho(\widehat{J}_{\mathit{min}}) \in I$ and $\rho(\widehat{J}_{\mathit{max}}) \in F$.
		\item
		For every consecutive $c_1 = (K_1, L_1),\, c_2 = (K_2, L_2) \in
		\widehat{J}$, i.e., such that $K_2 = K_1 \cup \{ j \}$ for some $j \in
		J$, one has $(\rho(c_1), w(j), \rho(c_2)) \in \Delta$.
		\item
		For every $c \neq \jmin\in \widehat{J}$ without any predecessor, one has $\lim_{c^-} \rho \rightarrow \rho(c)
		\in \Delta$, where $\lim_{c^-} \rho = \{ q \in Q \mid \forall c_1 < c\ .\ 
		\exists c_2\ .\ c_1 < c_2 < c \,\wedge\, \rho(c_2) = q \}$.
		\item
		For every $c \neq \jmax\in \widehat{J}$ without any successor, one has $\rho(c) \rightarrow \lim_{c^+}
		\rho \in \Delta$, where $\lim_{c^+} \rho = \{ q \in Q \mid \forall
		c_1 > c \ . \ \exists c_2\ . \ c < c_2 < c_1 \,\wedge\, \rho(c_2) = q \}$.
	\end{itemize}
\end{definition}

Intuitively, a run $\rho$ reading a word $w$ of length $J$ must assign an initial state to the first cut and a final state to the last cut. Moreover, if two cuts $c_1$ and $c_2$ are consecutive, there must exist a successor transition from $\rho(c_1)$ to $\rho(c_2)$, and this successor transition must read the symbol comprised between $c_1$ and $c_2$. The name \emph{successor transition} refers to the fact that $c_2$ is the successor of $c_1$ in the ordering $\widehat{J}$. On the other hand, if a cut $c$ (different from the last one) does not have any successor, then there exist cuts in every right neighborhood of $c$. To leave this cut, $\rho$ must follow a right-limit transition $\rho(c)\to P$. This is only valid if $P$ is exactly equal to the \emph{right-limit} of $\rho$ at cut $c$, i.e., if $P$ is equal to the set of states visited by $\rho$ in every right neighborhood of $c$. The case where a cut $c$ (different from the first cut) does not have any predecessor is symmetrical: such a cut must be reached by a left-limit transition. Note that the semantics of the limit transitions is close to the Muller acceptance condition in infinite-word automata~\cite{muller}, in the sense that a right-limit transition $q\to P$ (resp. a left-limit transition $P\to q$) can only be followed if $P$ is exactly equal to the set of states visited in every right neighborhood (resp. every left neighborhood) of a cut.

An automaton on linear orderings $\A$ \emph{accepts} a word $w$ if there exists a run reading $w$ in $\cal A$. The class of all such words forms the \emph{language accepted} by $\cal A$. A language is said to be \emph{rational} if it is accepted by an automaton. In this work, we focus on rational languages of words indexed by scattered (countable and uncountable) linear orderings. An automaton accepting such a language is an \emph{automaton on scattered linear orderings}. All the automata considered in the remainder of this article will be of this kind. 

A \emph{path} of $\A$ is similar to a run, except that it is not forced to start in an initial state and to end in a final one. We say that a word $w$ is \emph{read} by $\A$ if there exists a path reading it in $\A$. In what follows, we will often need to state that a path $\pi$ reading a word $w$ of length $J$ starts in a state $p\in Q$, ends in a state $q\in Q$ and \emph{visits exactly} the set of states $R\subseteq Q$, i.e., \[p=\pi(\jmin),\;R=\bigl\{\pi(c)\mid c\in \widehat{J}\;\bigr\},\;q=\pi(\jmax).\] We use the notation $\pi: p\xrightarrow{R}q$ to denote this concisely. 

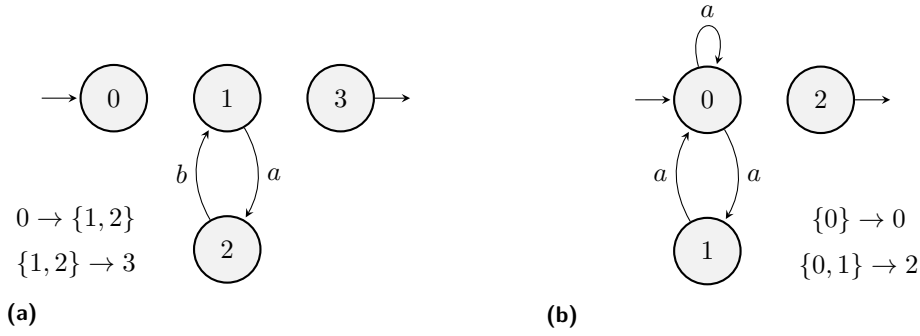
\begin{figure}[h]
	\centering
	\begin{subfigure}[h]{0.49\textwidth}
		\begin{tikzpicture}[on grid, auto, node distance = 2cm]
			\node[state, initial] (0) {$0$};
			\node[state] (1) [right = of 0, xshift=-0.5cm] {$1$};
			\node[state] (2) [below = of 1] {$2$};
			\node[state, accepting right] (3) [right = of 1, xshift=-0.5cm] {$3$};
			\node at (0, 1.2) {};
			\path
			(1) edge[bend left] node {$a$} (2)
			(2) edge[bend left] node {$b$} (1);
			\node at (-0.5, -1.6) {$0\to \{1,2\}$};
			\node at (-0.5, -2.2) {$\{1, 2\}\to 3$};
		\end{tikzpicture}
		\caption{}
		\label{fig:example-right-lim}
	\end{subfigure}
	\hfill
	\begin{subfigure}[h]{0.49\textwidth}
		\begin{tikzpicture}[on grid, auto, node distance = 2cm]
			\node[state, initial] (0) {$0$};
			\node[state] (1) [below = of 0] {$1$};
			\node[state, accepting right] (2) [right = of 0, xshift=-0.5cm] {$2$};
			\node at (0, 1.2) {};
			\node at (-2, 0) {};
			\path
			(0) edge[loop above] node {$a$} (0)
			(0) edge[bend left] node {$a$} (1)
			(1) edge[bend left] node {$a$} (0);
			\node at (2, -1.6) {$\{0\}\to 0$};
			\node at (2, -2.2) {$\{0, 1\}\to 2$};
		\end{tikzpicture}
		\caption{}
		\label{fig:example-cof-omega}
	\end{subfigure}
	\caption{Examples of automata on scattered linear orderings.}
\end{figure}
\begin{example}\label{ex:automata}
A state is marked initial with an incoming arrow without origin, and final with an outgoing arrow without destination. The automaton represented in Figure~\ref{fig:example-right-lim} accepts the bi-infinite word $(ab)^{-\omega}(ab)^\omega$ (since we do not distinguish isomorphic words, this is in fact the only accepted word). A run $\rho$ reading this word can be built as follows. For any cut $c$ immediately before a letter $a$, we set $\rho(c)=1$. Symmetrically, for any cut $c$ immediately before a letter $b$, we set $\rho(c)=2$. Finally, $\rho(\jmin)=0$ and $\rho(\jmax)=3$. We have $\lim_{\jmin^+}\rho=\{1,2\}=\lim_{\jmax^-}\rho$, hence the limit transitions $0\to \{1,2\}$ and $\{1,2\}\to 3$ can be followed to leave State $0$ and reach State $3$, respectively.
\end{example}
\begin{example}\label{ex:automata-uncountable}
The automaton depicted in Figure~\ref{fig:example-cof-omega} (adapted from~\cite{Bedon-ordinals}) accepts all the words $a^\alpha$ such that $\alpha$ is an ordinal (this follows directly from the absence of right-limit transitions) and $\cof(\alpha)= \omega$. Since the alphabet is unary in this case, we rather say that this automaton accepts all the ordinals of cofinality $\omega$. Such ordinals can be accepted by using the limit transition $\{0\}\to 0$ as many times as necessary, and by eventually following the limit transition $\{0,1\}\to 2$ to reach the last cut. However, if $\alpha$ does not have the cofinality $\omega$, then the limit transition $\{0, 1\}\to 2$ cannot be followed to reach State $2$:
\begin{itemize}
	\item If $\cof(\alpha)=0$, then $\alpha$ is empty. Hence it is not accepted (i.e., the empty word is not accepted).
	\item If $\cof(\alpha)=1$, then $\alpha$ has a greatest element. Therefore, $\alpha$ cannot be accepted since State $2$ does not have any incoming successor transition.
	\item If $\cof(\alpha)>\omega$, the situation is less intuitive. Assume that there exists a run $\rho$ accepting $\alpha$ with $\cof(\alpha)=\omega_1$. The arguments that follow also work if $\alpha$ has another uncountable cofinality (we will prove this formally in Section~\ref{sec:equivalence-relation}). Since State $2$ is the only final state and can only be reached by following $\{0,1\}\to 2$, one must have $\lim_{\alphamax^-}\rho=\{0,1\}$. The subordering $\widehat{\alpha}_1\subseteq \widehat{\alpha}$ of all the cuts mapped to $1$ by $\rho$ is thus  cofinal in $\widehat{\alpha}^*$. Given that $\cof(\widehat{\alpha}^*)=\omega_1$, this implies that $\cof(\widehat{\alpha}_1)=\omega_1$. 
	
	However, no transition can be taken from State $2$, hence the limit transition $\{0,1\}\to 2$ can be followed only once. This implies that $\widehat{\alpha}_1$ is isomorphic to a limit ordinal $\beta$ such that any $\gamma<\beta$ is not a limit ordinal. This leads to $\beta=\omega$, which is a contradiction since $\cof(\omega)=\omega\neq \omega_1$.
	\end{itemize}
The key element in this reasoning is the absence of a limit transition $\{0,1\}\to q$ with $q\in \{0,1\}$. In other words, after visiting infinitely often States $0$ and $1$, it is mandatory to exit $\{0,1\}$. The situation would be completely different if there was the additional limit transition $\{0, 1\}\to 0$, for instance. In this case, all the ordinals with a cofinality at least equal to $\omega$ would be accepted. It suggests that there are only two kinds of left-limit sets $P$ (the same classification applies to right-limit sets): those for which there exists $p\in P$ such that $P\to p\in \Delta$, and those that do not have this property and can only be used to read ordinals of cofinality $\omega$. The intuition is that automata on linear orderings are only able to distinguish two infinite cofinalities. We will formalize this idea in Sections~\ref{sec:constructible-words} and \ref{sec:equivalence-relation}, where we will prove that automata on scattered linear orderings are able to distinguish ordinals of cofinality greater than $\omega$ from ones of cofinality exactly equal to $\omega$, but cannot discriminate two different cofinalities beyond $\omega$. 

In fact, this example shows that automata on scattered linear orderings are not complementable. Indeed, there does not exist any automaton accepting exactly the orderings that are not ordinals, together with the ordinals of cofinality different from $\omega$. This comes from the observation that any automaton accepting an ordinal of cofinality greater than $\omega$ also accepts an ordinal of cofinality $\omega$. This property follows from~\cite[Theorem 2]{carton-emptiness}.
\end{example}
\section{Constructible Words}\label{sec:constructible-words}
We are now ready to define the notion of \emph{constructible} word over a given alphabet. Intuitively, a word is constructible if it can be obtained by applying some operators to the letters of the alphabet a finite number of times. Constructible words are convenient objects to work with, since they admit a finite representation (which is not the case for unrestricted words on scattered linear orderings, since there are uncountably many such words).
\begin{definition}\label{def:constructible}
	Let $\Sigma$ be a finite alphabet. The set of constructible words $\mathscr{C}$ on the alphabet $\Sigma$ is the minimal set of words (in the sense of inclusion) verifying the following properties:
	\begin{itemize}
		\item The words containing a single letter $\sigma\in \Sigma$ (i.e., the words indexed by a singleton) and the empty word $\varepsilon$ belong to $\mathscr{C}$.
		\item If $u,v\in \cons$, then $u\cdot v\in \cons$.
		\item If $u\in \cons$, then $u^\omega\in \cons$ and $u^{-\omega}\in\cons$.
		\item If $u\in \cons$, then $u^{\omega_1}\in \cons$ and $u^{-\omega_1}\in \cons$.
	\end{itemize}
\end{definition}
Notice that constructible words are a generalization of \emph{ultimately periodic words}, that are known to characterize $\omega$-regular languages~\cite{ultimately-periodic-words}. They also generalize the constructible words introduced by Carton to study the emptiness of automata on scattered linear orderings~\cite[Theorem 2]{carton-emptiness}. We show in Section~\ref{sec:constructible-words-characterization} that constructible words entirely characterize rational languages of words on scattered linear orderings. That is, two rational languages are equal if and only if they contain the same constructible words.

Observe that the only operators that are able to generate uncountable word lengths are the $\omega_1$-power and the reverse $\omega_1$-power. We will show in Section~\ref{sec:equivalence-relation} that uncountable cofinalities are not distinguishable by automata on scattered linear orderings (as already mentioned in Example~\ref{ex:automata-uncountable}).  Among all the uncountable cofinalities, we decided to choose the smallest one, i.e., $\omega_1$, as representative. Symmetrically, we chose $-\omega_1$ as representative of the reverse uncountable cofinalities.

\section{An Equivalence Relation on Words}\label{sec:equivalence-relation}
We have already mentioned several times that some words cannot be distinguished by an automaton on scattered linear orderings. We now formalize this by means of an equivalence relation on words. Intuitively, given an automaton $\A$, two words are equivalent if they can be read in the exact same ways in $\A$. This idea originates from Büchi's proof that $\omega$-regular languages are closed under complementation~\cite{Buc62,WolfgangThomas}. It was then adapted in~\cite{RispalFiniteRank} to automata on linear orderings (in the restricted case of finite-rank scattered words). 

Let $\A=(Q, \Sigma,\Delta,I,F)$ be an automaton on linear orderings. We say that two words $u$ and $v$ are \emph{equivalent for $\A$}, which we denote by $u\sim_{\A}v$, if for any $(p,R,q)\in Q\times 2^Q\times Q$, there exists a path $p\xrightarrow{R}q$ reading $u$ if and only if there exists a path $p\xrightarrow{R}q$ reading $v$.
When the automaton $\A$ is clear from the context, we also say that $u$ and $v$ are \emph{equivalent} and we write $u\sim v$ in place of $u\sim_\A v$.

The relation $\sim$ is an equivalence relation. The set of equivalence classes induced by $\sim$ is denoted by $S=\Sigma^{\textsf{scat}}/\sim$. We denote by $\varphi:\Sigma^{\textsf{scat}}\to S$ the function that maps a scattered word on the alphabet $\Sigma$ to its equivalence class. The equivalence class $\varphi(u)$ of a word $u$ can be entirely described by the set of all the triples $(p,R,q)\in Q\times 2^Q\times Q$ such that there exists a path $p\xrightarrow{R}q$ reading $u$. This set of triples is denoted by $\tilde\varphi(u)$, and can be considered as an unambiguous description of the equivalence class $\varphi(u)$. The set of words mapped to $s\in S$ by $\varphi$ is denoted by $\varphi^{-1}(s)$.

Since $Q$ is finite, there are only finitely many triples $(p,R,q)$. The number of sets of such triples is thus also finite, i.e., $S$ is necessarily a finite set.
\begin{example}\label{ex:equivalence}
	Consider again the automaton in Figure~\ref{fig:example-right-lim}. In this automaton, the words $\{(ab)^n\mid n\in \N_{>0}\}$ are all equivalent. For any $n>0$, we have \[\tilde\varphi\bigl((ab)^n\bigr)=\Bigl\{\bigl(1, \{1,2\},1\bigr)\Bigr\}\] since the word $(ab)^n$ can be read by a path $\pi$ if and only if $\pi$ is of the form $1\xrightarrow{\{1,2\}} 1$. Now consider the words $\{a^n\mid n\in \N_{>1}\}$. Again, all these words are equivalent. In fact, they cannot be read by any path, hence $\tilde\varphi(a^n)=\emptyset$ for any $n\in \N_{>1}$. Finally, notice that in this case, the language accepted by this automaton, which consists of the unique word $(ab)^{-\omega}(ab)^\omega$, is equal to $\varphi^{-1}(s)$ where $s$ is the equivalence class described by $\{(0, \{0,1,2,3\}, 3)\}$.
	
	However, given an arbitrary equivalence class $s\in S$, the language $\varphi^{-1}(s)$ is not rational in general. Consider the automaton in Figure~\ref{fig:example-cof-omega}. The language $L$ of all the words $u$ on the alphabet $\{a\}$ such that 
	\[\label{eq:ex2}
	\tilde\varphi(u)=\Bigl\{\bigl(0, \{0\}, 0\bigr),\bigl(0, \{0,1\}, 0\bigr),\bigl(1, \{0,1\}, 0\bigr) \Bigr\}
	\]
is not rational. Indeed, it contains all the words $a^\alpha$ where $\alpha$ is an ordinal such that:
	\begin{itemize}
		\item $\cof(\alpha)\neq 0$ (since the empty word cannot be read by a path visiting $\{0,1\}$),
		\item $\cof(\alpha)\neq 1$ (since any word with a last symbol can be read by a path $\pi: 0\xrightarrow{\{0,1\}} 1$, but $(0, \{0,1\}, 1)$ does not belong to $\tilde\varphi(u)$),
		\item and $\cof(\alpha)\neq \omega$ (given that any word $a^\alpha$ with $\cof(\alpha)=\omega$ can be read by a path $\pi: 0\xrightarrow{\{0,1,2\}} 2$, and $(0, \{0,1,2\}, 2)$ does not belong to $\tilde\varphi(u)$).
	\end{itemize}	
In fact, $L$ is equal to the set of words $a^\alpha$ where $\alpha$ is an ordinal such that $\cof(\alpha)>\omega$. As previously mentioned in Example~\ref{ex:automata-uncountable}, this language is not rational.
\end{example}

We are now ready to formalize the fact that automata on scattered linear orderings are not able to discriminate uncountable cofinalities. This is the purpose of Lemma~\ref{lemma:ordinal-powers-indistinguishable}. Given two regular uncountable ordinals $\alpha$ and $\beta$ and an automaton on scattered linear orderings $\A$, we prove that if there exists a path $\pi$ reading a word $w^\alpha$ in $\A$, then one can build another path $\pi'$ that reads $w^\beta$ and starts in the same state, ends in the same state and visits exactly the same states as $\pi$. Our strategy is to cut $\pi$ into pieces (Proposition~\ref{proposition:ordinal-limit-transitions}), to minimize them (Proposition~\ref{prop:countable-is-enough-for-reachability}) and finally to recompose them to build $\pi'$ (Proposition~\ref{prop:ordinal-powers-indistinguishable-prelim}). We start by two propositions related to ordinals. Their proofs are standard and we only give them for the sake of self-containment. They can be adapted to deal with reverse ordinals, by replacing the cofinality by the coinitiality.
\begin{proposition}\label{prop:limit-ordinals-cofinal-in-ordinal-of-uncountable-cofinalities}
	Let $\alpha$ be an ordinal with uncountable cofinality, i.e., with a cofinality at least equal to $\omega_1$. Then the set of limit ordinals strictly lower than $\alpha$ is cofinal in $\alpha$.
\end{proposition}
\begin{proof}
	Let $\beta<\alpha$ be an ordinal. We need to find a limit ordinal $\gamma$ such that $\beta< \gamma < \alpha$. First, notice that it is possible to find $\delta$ such that $\beta < \delta < \alpha$, since $\alpha$ is a limit ordinal. If $\delta$ is a limit ordinal, we can set $\gamma=\delta$ and the proposition is proven correct. If $\delta$ is a successor ordinal, then consider the sequence \[\delta_0=\delta,\; \delta_1=\delta+1,\;\delta_2=\delta+2,\;\dots\]
	and define $\gamma$ as the least ordinal greater than $\delta_n$ for all $n\in \N$. This definition guarantees that $\gamma$ is a limit ordinal. Furthermore, $\gamma <\alpha$ since $\alpha$ has an uncountable cofinality (indeed, if $\gamma>\alpha$, then $\alpha$ would have the cofinality $1$, and if $\gamma=\alpha$, then $\gamma$ would have the cofinality $\omega$).
 \end{proof}
\begin{proposition}\label{prop:cofinal-subset-of-a-regular-ordinal}
	If $\alpha$ is a regular ordinal, then any subset of $\alpha$ cofinal in $\alpha$ is isomorphic to $\alpha$.
\end{proposition}
\begin{proof}
	Assume that $\beta\subseteq \alpha$ is cofinal in $\alpha$ but not isomorphic to $\alpha$. It is impossible for $\beta$ to be isomorphic to an ordinal strictly greater than $\alpha$. Hence $\beta$ is isomorphic to $\gamma<\alpha$. Subsequently, the cofinality of $\beta$ is strictly smaller than $\alpha$ (since a regular ordinal is the smallest ordinal of its cofinality). This contradicts the fact that $\alpha$ and $\beta$ must have the same cofinality, given that $\beta$ is cofinal in $\alpha$.
 \end{proof}

Given a word $w$ and an ordinal $\alpha$, we say that a cut of the word $w^\alpha$ is \emph{external} (that is, external to $w$) if it separates $w^\alpha$ into two factors $w^{\alpha_1}$ and $w^{\alpha_2}$ (with $\alpha_1$ and $\alpha_2$ two ordinal numbers, possibly equal to $0$, and such that $\alpha=\alpha_1+\alpha_2$).
	\begin{proposition}\label{proposition:ordinal-limit-transitions}
		Let $\A$ be an automaton on linear orderings and let $\pi$ be a path reading a word $w^\alpha$ in $\A$, where $w$ is a non-empty word and $\alpha$ is an uncountable regular ordinal. Let $P\to q$ be the limit transition followed by $\pi$ to reach the last cut of $w^\alpha$. Then, there exists a state $p$ and an ordering $(c_i)_{i\in \alpha}$ of external cuts of $w^\alpha$ such that:
		\begin{itemize}
			\item $p\in P$ and $P\to p$ is a limit transition of $\A$,
			\item $\pi(c_i)=p$ for any $i\in \alpha$,
			\item and for any $i<j\in \alpha$, $\pi$ visits $P$ between $c_i$ and $c_j$,
		\end{itemize}
	\end{proposition}
	\begin{proof}
		Let $J$ be the length of $w^\alpha$. The idea of the proof is to build a chain of suborderings of $\widehat{J}$ that are all cofinal in $\widehat{J}^*$ and isomorphic to $\alpha$. The last element of this chain will be the ordering $(c_i)_{i\in \alpha}$.
		
		First, let $C\subseteq \widehat{J}^*$ be the set of non-extremal cuts external to $w$. The ordering $C$ is cofinal in $\widehat{J}^*$ and isomorphic to $\alpha$. Next, define $C^P\subseteq C$ as an ordering cofinal in $C$ and such that $\pi$ visits exactly $P$ between any two cuts in $C^P$. The existence of this ordering and the fact that it is cofinal in $C$ follow from the assumption that $\pi$ reaches the final cut of $w^\alpha$ thanks to the limit transition $P\to q$. Given that $C^P$ is cofinal in an ordering that is isomorphic to the regular ordinal $\alpha$, $C^P$ is itself isomorphic to $\alpha$ (thanks to Proposition~\ref{prop:cofinal-subset-of-a-regular-ordinal}).
		
		Secondly, define the ordering $C^{P,\mathit{lim}}\subseteq C^P$ as the set of cuts in $C^P$ that do not have a predecessor in $C^P$. By Proposition~\ref{prop:limit-ordinals-cofinal-in-ordinal-of-uncountable-cofinalities}, $C^{P,\mathit{lim}}$ is cofinal in $C^P$ (since the cofinality of $C^P$ is $\alpha$, i.e., is uncountable). By Proposition~\ref{prop:cofinal-subset-of-a-regular-ordinal}, $C^{P,\mathit{lim}}$ is thus isomorphic to $\alpha$. Figure~\ref{fig:example-C^P} shows some cuts of $C^P$. Between any two such cuts, $\pi$ visits exactly $P$. The bigger vertical bar represents a cut in $C^{P,\mathit{lim}}$.  
		\begin{figure}[h!]
			\centering
			\begin{tikzpicture}
				\node at (5.25, 0) {$\dots$};
				 
				\foreach \i in {-1,...,1} {
					\node at (\i+3,0) {};
					\fill (\i+3-0.5,-0.2) rectangle (\i+3-0.45, 0.2);
					\path (\i+3-0.5+0.1,0) edge node[above] {$P$} (\i+1+3-0.5-0.05, 0);
					
				}
				\fill (2+3-0.5,-0.2) rectangle (2+3-0.45, 0.2);
				\fill (0+6,-0.4) rectangle (0.05+6, 0.4);

			\end{tikzpicture}
			\caption{Several cuts of $C^P$ and a cut of $C^{P,\mathit{lim}}$.}
			\label{fig:example-C^P}
		\end{figure}
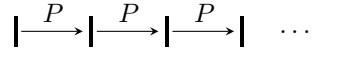
		
		Each of the cuts in $C^{P,\mathit{lim}}$ (except maybe the first one) is reached by $\pi$ thanks to a left-limit transition involving $P$ as limit set. There are finitely many such limit transitions. Hence there exists a state $p$ and an ordering $C^{P\to p}\subseteq C^{P,\mathit{lim}}$ cofinal in $C^{P,\mathit{lim}}$ (and thus isomorphic to $\alpha$ by Proposition~\ref{prop:cofinal-subset-of-a-regular-ordinal}) such that $\pi$ follows the limit transition $P\to p$ to reach each cut in $C^{P\to p}$. Note that we necessarily have $p\in P$, since the set of states visited between two cuts in $C^{P\to p}$ is exactly $P$. Therefore, the ordering $C^{P\to p}$ satisfies all the properties enumerated in the proposition.
		\end{proof}
	
\begin{proposition}\label{prop:countable-is-enough-for-reachability}
	Let $\A$ be an automaton on scattered linear orderings with set of states $Q$. Let $\alpha\neq 0$ be an ordinal number and let $w$ be a word. Assume that there exists a path $\pi:p\xrightarrow{R}q$ reading $w^\alpha$ in $\A$, where $p,q\in Q$ and $R\in 2^Q$. Then, there exists a path $\pi': p\xrightarrow{R}q$ and a countable ordinal number $\alpha'\neq 0$ such that $\pi'$ reads $w^{\alpha'}$ in $\A$.
\end{proposition}
\begin{proof}
	We claim that there exists an automaton $\B$ that accepts a word indexed by the ordinal $\beta$ if and only if $w^\beta$ can be read in $\A$ by a path starting in $p$, ending in $q$ and visiting exactly $R$. Consider the following candidate $(Q_\B, \Sigma_\B, \Delta_\B, I_\B, F_\B)$ for $\B$: 
	\begin{itemize}
		\item $Q_\B=Q\times 2^Q\times 2^Q$. Intuitively, the first component of a state $q_\B\in Q_\B$ tracks the current state of $\A$, the second one tracks the set of states of $\A$ already visited (which by convention includes the current state), and the last one tracks the set of states of $\A$ that needs to be visited during the next reading of $w$ (which necessarily exists, unless the entire word $w^\beta$ has already been read in $\A$). 
		\item $\Sigma_\B=\{\sigma\}$. The symbol $\sigma$ replaces the reading of $w$ in $\A$. 
		\item The successor and limit transitions of $\B$ enforce the semantics of $Q_\B$ and $\Sigma_\B$. Formally, they are defined as follows:
		\begin{itemize}
			\item Given $(q_1,Q_1,U_1), (q_2, Q_2,U_2)\in Q_\B$ such that $q_1\in Q_1$ and $q_2\in Q_2$, the successor transition from $(q_1, Q_1,U_1)$ to $(q_2, Q_2,U_2)$ belongs to $\Delta_\B$ if and only if the word $w$ can be read in $\A$ by a path starting in $q_1$, ending in $q_2$ and visiting exactly $U_1$, with the additional requirement that $Q_1\cup U_1=Q_2$. 
			\item Given $\{(q_1,Q_1,U_1), (q_2, Q_2,U_2),\dots, (q_\ell, Q_\ell,U_\ell)\}\subseteq Q_\B$ and $(q_{\ell+1},Q_{\ell+1},U_{\ell+1})\in Q_\B$, the left-limit transition \[\Bigl\{(q_1,Q_1,U_1), (q_2, Q_2,U_2),\dots, (q_\ell, Q_\ell,U_\ell)\Bigr\}\to (q_{\ell+1},Q_{\ell+1}, U_{\ell+1})\] belongs to $\Delta_\B$ if and only if the left-limit transition $\{q_1,q_2,\dots,q_\ell\}\to q_{\ell+1}$ exists in $\A$, with the additional requirements that $\{q_{\ell+1}\}\cup\bigcup_{1\leq i\leq \ell} Q_i=Q_{\ell+1}$ (to satisfy the semantics of $Q_{\ell+1}$) and $U_i\subseteq \{q_1,q_2,\dots,q_\ell\}$ for any $1\leq i\leq \ell$ (to make sure that no reading of $w$ in $\A$ visits a state that is not in the limit set $\{q_1,q_2,\dots,q_\ell\}$).
		\end{itemize}
		\item $I_\B=\{(p, \{p\}, \{p\}\cup U)\mid U\in 2^Q\}$ and $F_\B=\{(q, R, U)\mid U\in 2^Q\}$. 
	\end{itemize}
	The semantics of the states, the choice of initial and final states and the selection of the successor transitions guarantees that $\sigma^\beta$ is accepted by $\B$ if and only if $w^\beta$ can be read in $\A$ by a path starting in $p$, ending in $q$ and visiting $R$.
	
	By assumption, there exists a path $\pi:p\xrightarrow{R}q$ reading the word $w^\alpha$ in $\A$. Therefore, $\B$ accepts $\sigma^\alpha$. The language accepted by $\B$ is thus non-empty. Following the principles of the emptiness test developed in \cite{carton-emptiness}, one can find a word accepted by $\B$ that has finite rank. This word takes the form $\sigma^{\alpha'}$ where $\alpha'$ is an ordinal less than $\omega^\omega$, hence countable. Finally, by definition of $\B$, the word $w^{\alpha'}$ can be read in $\A$ by a path $\pi':p\xrightarrow{R}q$.
	 \end{proof}

	\begin{proposition}\label{prop:ordinal-powers-indistinguishable-prelim}
		Let $\alpha$ and $\beta$ be two regular uncountable ordinals. Let $w$ be a scattered word and let $\A$ be an automaton on scattered linear orderings with set of states $Q$. If there exists a path $\pi:q_1\xrightarrow{R}q_2$ reading $w^\alpha$, then there also exists a path $\pi':q_1\xrightarrow{R}q_2$ reading $w^\beta$.
	\end{proposition}
	\begin{proof}
        We denote by $J$ the length of $w$. If $w$ is the empty word, $w^\alpha=w^\beta$, hence the proposition is proven correct. We thus assume that $w$ is not the empty word. The idea of the proof is to analyze $\pi$, and then to build a path $\pi'$ reading $w^\beta$ that keeps the properties of $\pi$.
		
		First, notice that since $\alpha$ is regular and uncountable, it is an initial ordinal greater or equal to $\omega_1$. Therefore it is necessarily a limit ordinal. The path $\pi$ thus reaches the last state $q_2$ via a limit transition $P\to q_2$ where $P\subseteq Q$. Then, there exists a state $p\in P$ and a sequence of cuts $(c_i)_{i\in \alpha}$ with the properties enumerated in Proposition~\ref{proposition:ordinal-limit-transitions}. Let $\jmin$ be the first cut of $w^\alpha$. We assume without loss of generality that $c_0\neq \jmin$ (if $c_0=\jmin$, then redefine the sequence $(c_i)_{i\in \alpha}$ by setting $c_0$ to $c_1$, $c_1$ to $c_2$, etc.). Let us define $\pi_1$ as the restriction of $\pi$ to the prefix $w[\jmin:c_0]$, and $\pi_2$ as the restriction of $\pi$ to the subword $w[c_0:c_1]$. From Proposition~\ref{proposition:ordinal-limit-transitions}, we infer the following properties of $\pi_1$ and $\pi_2$ (together with the fact that $P\to p$ is a limit transition of $\A$):
		\begin{itemize}
			\item $\pi_1$ reads $w^{\theta_1}$ for some ordinal number $\theta_1>0$, starts in $q_1$, ends in $p$ and visits exactly some set of states $U\subseteq Q$,
			\item $\pi_2$ reads $w^{\theta_2}$ for some ordinal number $\theta_2>0$, starts in $p$, ends in $p$ and visits exactly $P$,
			\item and $U\cup P \cup \{q_2\} = R$.
		\end{itemize}
		Now let us apply Proposition~\ref{prop:countable-is-enough-for-reachability} twice (once for $\pi_1$ and once for $\pi_2$). This yields a path $\pi_1': q_1\xrightarrow{U} p$ reading a word $w^{\theta_1'}$ and a path $\pi_2': p\xrightarrow{P}p$ reading a word $w^{\theta_2'}$, where $\theta_1',\theta_2'$ are countable ordinal different from $0$.
		
		We are now ready to build a path $\pi':q_1\xrightarrow{R}q_2$ reading $w^\beta$. First, notice that $\beta>\theta_1'>0$ and $\beta>\theta_2'>0$ (since $\beta$ is uncountable and $\theta_1',\theta_2'$ are countable), and $\beta$ is an initial ordinal (since it is regular). This implies that $\theta_2'\cdot \beta=\beta$ and $\theta_1'+\beta=\beta$ (see e.g.~\cite[Exercise 3.39]{Ros82}). In terms of words, this means that 
		\[w^\beta=w^{\theta_1'}\cdot (w^{\theta_2'})^\beta.\]
		Now consider the path $\pi'$ that first follows $\pi_1'$ to read $w^{\theta_1'}$, then follows $\beta$ times $\pi'_2$ to read each copy of $w^{\theta_2'}$, and eventually reaches the last cut of $w^\beta$ via the limit transition $P\to q_2$. This path reads $w^\beta$, starts in $q_1$, visits $R$ and ends in $q_2$. Moreover, it can be checked that this path is valid. Any cut of the word $(w^{\theta_2'})^\beta$ that is external to $w^{\theta_2'}$ and that does not have a predecessor that is external can be reached thanks to the limit transition $P\to p$ (or $P\to q_2$ for the last cut). 
	 \end{proof}
	\begin{lemma}\label{lemma:ordinal-powers-indistinguishable}
		Let $\alpha$ and $\beta$ be two regular uncountable ordinals. For any automaton on scattered linear orderings $\A$ and for any scattered word $w$, the words $w^\alpha$ and $w^\beta$ are equivalent for $\A$. Symmetrically, $w^{-\alpha}$ and $w^{-\beta}$ are also equivalent for $\A$.
	\end{lemma}
	\begin{proof}
		It follows from the fact that Proposition~\ref{prop:ordinal-powers-indistinguishable-prelim} is symmetric regarding $\alpha$ and $\beta$. Hence, the set of triples $(q_1,R,q_2)\in Q\times 2^Q\times Q$ that characterize the equivalence class of $w^\alpha$ are the same as those characterizing the equivalence class of $w^\beta$, i.e., $w^\alpha$ and $w^\beta$ are in the same equivalence class.
		
		The proof that $w^{-\alpha}$ and $w^{-\beta}$ are equivalent for $\A$ is essentially the same. It relies on the mirrored versions of Propositions~\ref{prop:limit-ordinals-cofinal-in-ordinal-of-uncountable-cofinalities} to \ref{prop:ordinal-powers-indistinguishable-prelim}.
	 \end{proof}

The next step is to understand how the equivalence classes induced by $\sim_\A$ behave with respect to the concatenation of words. Lemma~\ref{lemma:equivalence-relation} addresses this question. It will be used intensively in Section~\ref{sec:constructible-words-characterization}.

\begin{lemma}\label{lemma:equivalence-relation}
	Let $\A$ be an automaton on scattered linear orderings and let $I$ be a non-empty scattered linear ordering. Let $\{u_i\mid i\in I\}$ and $\{v_i\mid i\in I\}$ be two sets of words such that $u_i\sim_{\cal A} v_i$ for all $i\in I$. Then the words $u$ and $v$ respectively defined as $\bigl(\prod_{i \in I}u_i\bigr)$ and $ \bigl(\prod_{i \in I} v_i\bigr)$ are equivalent for $\A$.
\end{lemma}
\begin{proof}
	In this proof, given two paths $\pi$ and $\tau$ in $\A$, we write $\pi\equiv\tau$ if both paths start in the same state, end in the same state, and visit exactly the same set of states. Moreover, we denote by $J_i$ the length of $v_i$, for any $i\in I$.
	
	Since $u_i\sim_\A v_i$, the existence of a path $\pi_i$ reading $u_i$ in $\A$ implies that there also exists a path $\tau_i\equiv \pi_i$ reading $v_i$ in $\A$. The idea of the proof is thus to consider a path $\pi$ reading $u=\prod_{i\in I}u_i$ in $\A$ and then to split it into a set of paths $\{\pi_i\mid i\in I\}$ such that $\pi_i$ reads $u_i$. From this, we can infer the existence of a set of paths $\{\tau_i\mid i\in I\}$ such that $\tau_i$ reads $v_i$ and $\tau_i\equiv \pi_i$, and finally build a path $\tau\equiv \pi$ reading $v$ in $\A$. By symmetry of the problem, this is sufficient to show that $u$ and $v$ are equivalent for $\A$.
	
	We now build explicitly the path $\tau$ and show that this path is valid. To do so, we give the value of $\tau(c)$ for each cut $c$ of $v$, depending on the properties of $c$:
	\begin{itemize}
		\item If $c$ is internal to some word $v_i$ (that is, if $c\in \widehat{J_i}^*$ for some $i\in I$), then $\tau(c)=\tau_i(c)$.
		\item Otherwise, $c$ belongs to the set $C$ of the external cuts of $v$ (that is, the cuts that are not internal to some $v_i$). In fact, $C$ can be seen as a linear ordering isomorphic to $\widehat{I}$. We distinguish four cases:
		\begin{itemize}
			\item $c$ has a successor and a predecessor in $C$. In that case, $c$ is the last cut of $J_i$ and the first cut of $J_{i+1}$, where $i$ and $i+1$ are two consecutive elements of $I$. Then we set $\tau(c)=\tau_i(c)=\tau_{i+1}(c)$. Note that this last equality holds because the last state of $\tau_i$ is equal to the first state of $\tau_{i+1}$ (given that the last state of $\pi_i$ is equal to the first state of $\pi_{i+1}$).
			\item $c$ has a successor but no predecessor in $C$. In that case, $c$ is the first cut of $J_i$ for some $i\in I$, and $i$ does not have any predecessor in $I$. We set $\tau(c)=\tau_i(c)$.
			\item Symmetrically, if $c$ has a predecessor but no successor in $C$, then $c$ is the last cut of $J_i$ for some $i\in I$, and $i$ does not have any successor in $I$. As in the previous case, we set $\tau(c)=\tau_i(c)$.
			\item Finally, if $c$ does not have a predecessor nor a successor in $C$, then there does not exist any $i$ such that $c\in \widehat{J_i}$. Therefore, we simply set $\tau(c)=\pi(c)$.
		\end{itemize}
		Notice that in fact, we set $\tau(c)=\pi(c)$ for any $c\in C$. The above discussion might seem unnecessary but is actually meant to prove in an easier way that $\tau$ is a valid path. Given that each $\tau_i$ is a valid path for any $i\in I$, we only need to check that an external cut $c$ of $v$ that does not have any predecessor (resp. successor) can be reached by a left-limit transition (resp. a right-limit transition). This is always the case, since each path $\tau_i$ visits exactly the same states as $\pi_i$, for all $i\in I$. Therefore, $\lim_{c^-}\tau=\lim_{c^-}\pi$ and $\lim_{c^+}\tau=\lim_{c^+}\pi$ for any external cut $c$. The path $\tau$ can thus reach (resp. leave) cut $c$ with the same left-limit transition (resp. right-limit transition) as $\pi$, given that $\tau(c)=\pi(c)$.
	\end{itemize}
    This concludes the proof.
\end{proof}
In the particular case where $\alpha=2$, Lemma~\ref{lemma:equivalence-relation} actually shows that it makes sense to define a binary product operation between the equivalence classes induced by $\sim_\A$. Given two equivalence classes $s_1$ and $s_2$, their product $s_1\cdot s_2$ is defined as the equivalence class of any word $u_1\cdot u_2$ such that $u_1$ belongs to $s_1$ and $u_2$ to $s_2$. Equivalently,
\begin{equation}\label{eq:phi-morphism}\varphi(u_1)\cdot \varphi(u_2)=\varphi(u_1\cdot u_2).\end{equation}
for any $u_1,u_2\in \Sigma^{\textsf{scat}}$.
Notice that the product we defined on $S$ is an associative operation, since
\[\varphi(u_1)\cdot \bigl(\varphi(u_2)\cdot \varphi(u_3)\bigr)=\varphi(u_1\cdot u_2\cdot u_3)=\bigl(\varphi(u_1)\cdot \varphi(u_2)\bigr)\cdot \varphi(u_3).\]
Therefore, the set $S$ containing the equivalence classes induced by $\sim_\A$ has a finite-semigroup structure, since it is a finite set equipped with a binary product operation that is associative. Moreover, it follows from Equation~\ref{eq:phi-morphism} that $\varphi:\Sigma^\textsf{scat}\to S$ is a \emph{semigroup morphism} from $\Sigma^\textsf{scat}$ (which can be seen as an infinite semigroup where the associative operation is the concatenation) to the finite semigroup $S$.

The links between automata and semigroups are well studied in the literature. In~\cite{Ris05}, rational languages of words indexed by countable scattered linear orderings are shown to be exactly the languages \emph{recognizable} by a special form of finite semigroups. That is, a language $L\subseteq \Sigma^\diamond$ (where $\Sigma^\diamond$ is the class of all the words on the alphabet $\Sigma$ indexed by countable scattered linear orderings) is rational if and only if it can be written as $\bigcup_{t\in T}\varphi^{-1}(t)$ where $T$ is a subset of a finite semigroup $S$ and $\varphi:\Sigma^\diamond\to S$ is a semigroup morphism (i.e., $\varphi$ satisfies Equation~\ref{eq:phi-morphism}).
This result generalizes the correspondence between regular languages of finite words and finite semigroups~\cite{eilenberg}. It also proves that rational languages of words indexed by countable scattered linear orderings are closed under complementation.

These results do not extend to rational languages of scattered words (countable and uncountable). Given a finite semigroup $S$ and a semigroup morphism $\varphi:\Sigma^{\textsf{scat}}\to S$, the languages $\varphi^{-1}(s)$ (where $s\in S$) are not rational in general, hence the equivalence between rationality and recognizability by semigroups does not hold in our settings. If it were the case, then automata on scattered linear orderings would be complementable. Nevertheless, it is still possible to exploit the fact that the set of equivalence classes induced by $\sim_\A$ has a finite-semigroup structure. This is the objective of Section~\ref{sec:constructible-words-characterization}.

\section{Constructible Words Characterize Rational Languages of Scattered Words}\label{sec:constructible-words-characterization}

In this section, we prove that the language accepted by an automaton on scattered linear orderings $\A$ is characterized by the constructible words accepted by $\A$ (Theorem~\ref{thm:constructible-words-scattered}). We first prove Theorem~\ref{lemma:an-equivalence-class-contains-a-constructible-word}, stating that for any non-empty equivalence class $s$ induced by $\sim_\A$, the language $\varphi^{-1}(s)$ contains a constructible word (despite being non-rational in general). The proof of this result relies on the application of a theorem of Colcombet~\cite{Colcombet-ramseyan-splits}. We introduce this theorem in Section~\ref{sec:colcombet} and apply it to our problem in Section~\ref{sec:application}.

\begin{remark}
    In the restricted case of automata on countable scattered linear orderings, Theorem~\ref{lemma:an-equivalence-class-contains-a-constructible-word} follows immediately from the fact that each equivalence class is a rational language~\cite[Proposition 27]{Ris05}, and any non-empty rational language contains a countable constructible word~\cite[Theorem 2]{carton-emptiness}. Alternatively, the fact that the semigroups introduced in~\cite{Ris05} (as an algebraic model equivalent to automata on countable scattered linear orderings) can be described by an algebra whose operators are the concatenation, the $\omega$-power and the reverse $\omega$-power, yields another proof of Theorem~\ref{lemma:an-equivalence-class-contains-a-constructible-word} in the restricted case of countable scattered words. In this restricted case, automata can be complemented~\cite{Ris05}, hence Theorem~\ref{thm:constructible-words-scattered} can also be directly proven by adapting the proof of the characterization of the $\omega$-regular languages by their ultimately periodic words~\cite[Fact 1]{ultimately-periodic-words}.
\end{remark}
\subsection{A Ramsey-like Theorem for Semigroups}\label{sec:colcombet}
The theorem presented in this section comes from \cite{Colcombet-ramseyan-splits}. It can be seen as an extension of Ramsey's theorem~\cite{Ramsey} to graphs whose edges are colored according to a semigroup structure. Its proof relies on the axiom of choice.

Let $J$ be a linear ordering and let $S$ be a finite semigroup. The product of two elements $s_1,s_2\in S$ is denoted by $s_1\cdot s_2$. A \emph{multiplicative labeling} $\sigma$ is a function that maps each pair $(i,j)\in J\times J$ such that $i<j$ to an element of $S$, and such that
\[\sigma(i,j)\cdot \sigma(j,k)=\sigma(i,k)\]
for all $i<j<k\in J$. In other words, if we consider $J$ as an ordered set of vertices of a complete graph, the multiplicative labeling $\sigma$ is a coloring of the edges of this graph, where the colors satisfy a semigroup structure (which is not the case in the classical infinite form of Ramsey's theorem~\cite{Ramsey}). 

A \emph{split} of $J$ is a function $r:J\to \{1,2,\dots, N\}$ where $N\in \N$ is called the \emph{height} of the split. Given a split $r$ and an integer $m\in \{1,2,\dots,N\}$, two elements $i,j\in J$ are said to be \emph{$m$-neighbors} if $r(i)=r(j)=m$ and $r(k)\leq m$ for all $k$ between $i$ and $j$. The relation ``being $m$-neighbors'' is in fact an equivalence relation, for any $m\in \{1,2,\dots, N\}$. An \emph{$m$-class} is an equivalence class induced by this equivalence relation. A split $r:J\to \{1,2,\dots, N\}$ is \emph{Ramseyan} for the multiplicative labeling $\sigma$ if for any $m\in \{1,2,\dots, N\}$ and for any $m$-class $X$, there exists an idempotent element $e\in S$ such that $\sigma(i,j)=e$ for all $i<j\in X$ (an element $e\in S$ is \emph{idempotent} if $e\cdot e=e$). 
\begin{theorem}[Colcombet \cite{Colcombet-ramseyan-splits}]\label{thm:colcombet}
    For any linear ordering $J$, for any finite semigroup $S$ and for any multiplicative labeling $\sigma:J\times J\to S$, there exists a split of $J$ of height at most $2|S|$ that is Ramseyan for $\sigma$. 
\end{theorem}
\begin{example}
	The following example comes from~\cite{Colcombet-ramseyan-splits}. Consider the graph with $J=17$ vertices represented in Figure~\ref{fig:ramseyan-split}.
	\begin{figure}[h!]
		\centering
		\scalebox{0.75 }{
			\begin{tikzpicture}
				\foreach \i in {1,...,17} {
					\node[circle, fill=black, inner sep=2pt] (P\i) at (\i,0) {};
				}
				
				\foreach \i in {3,8,9,13,14,15} {
					\draw[blue,bend left=40] (P\i) to (P\the\numexpr\i+1\relax);
				}
				
				\foreach \i in {2} {
					\draw[red,bend left=40] (P\i) to (P\the\numexpr\i+1\relax);
				}
				
				\foreach \i in {1,4,6,11,12, 16} {
					\draw[olive,bend left=40] (P\i) to (P\the\numexpr\i+1\relax);
				}
				
				\foreach \i in {5,7,10} {
					\draw[orange,bend left=40] (P\i) to (P\the\numexpr\i+1\relax);
				}
				
				\foreach \i in {3,8,9,13,14,15} {
					\node[blue] at (\i+0.5, 0.6) {\Large $0$};
				}
				
				\foreach \i in {2} {
					\node[red] at (\i+0.5, 0.6) {\Large $1$};
				}
				
				\foreach \i in {1,4,6,11,12, 16} {
					\node[olive] at (\i+0.5, 0.6) {\Large $2$};
				}
				
				\foreach \i in {5,7,10} {
					\node[orange] at (\i+0.5, 0.6) {\Large $3$};
				}
				
				\foreach \i in {1,5,7,13,14,15,16} {
					\node at (\i,-0.5) {\Large $3$};
				}

				\foreach \i in {3,4,6,8,9,10,12,17} {
					\node at (\i,-0.7) {\Large $2$};
				}
				
				\foreach \i in {2,11} {
					\node at (\i,-0.9) {\Large $1$};
				}
				
			\node at (0, -0.6) {\Large $r:$};
			
		\end{tikzpicture}}
	\caption{Example of Ramseyan split.}
	\label{fig:ramseyan-split}
	\end{figure}
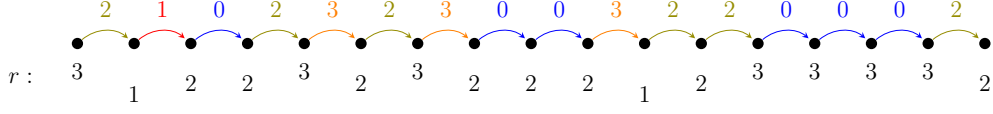
The edges are labeled by the elements of the semigroup $S=\mathbb{Z}/5\mathbb{Z}$ (the set of integers modulo $5$, equipped with addition modulo $5$, that we denote by $+$ here). This semigroup is in fact a group, but it does not matter in this example. Note that all the edges are not represented, but the label of any edge can be retrieved thanks to the multiplicative property of the labeling. For instance, the edge linking the first vertex to the third one is labeled by $2+1=3$.
	
Theorem~\ref{thm:colcombet} predicts the existence of a Ramseyan split of height at most $2|S|=10$. In fact, in this case there exists a Ramseyan split $r$ of height $3$, represented below the graph in Figure~\ref{fig:ramseyan-split}. Indeed, consider a $3$-class (there is only one such class, which is the set of vertices mapped to $3$). Two distinct vertices in this class are always linked by an edge labeled by the idempotent element $0$. Now consider a $2$-class (this time, there exist $4$ such classes, two of them being a singleton). Again, two distinct elements in a $2$-class are linked by an edge labeled by $0$ (in fact, $0$ is the only idempotent element of the semigroup). Finally, there are two $1$-classes that are both singletons. Hence one cannot find two distinct vertices in these $1$-classes, which implies that the Ramseyan condition is automatically satisfied.
	\end{example}
	
\subsection{Application to Constructible Words}\label{sec:application}
Let $\A=(Q, \Sigma, \Delta, I, F)$ be an automaton on linear orderings and let $S$ be the set of equivalence classes induced by $\sim_\A$. We regard $S$ as a finite semigroup, as established in Section~\ref{sec:equivalence-relation}. We denote by $\varphi:\Sigma^{\textsf{scat}}\to S$ the semigroup morphism mapping a scattered word to its equivalence class.

Let $w$ be a word indexed by a linear ordering $J$. We now build a multiplicative labeling of $\widehat{J}$ from the equivalence classes of each subword of $w$ (this idea is already present in~\cite{Colcombet-ramseyan-splits}). More precisely, we define the multiplicative labeling $\sigma:\widehat{J}\times \widehat{J}\to S$ as follows:
\[\sigma(c_1,c_2)=\varphi(w[c_1:c_2])\]
for any $c_1<c_2\in \widehat{J}$. Thanks to Theorem~\ref{thm:colcombet}, we know the existence of a split $r:\widehat{J}\to \{1,2,\dots,N\}$ whose height $N$ is at most $2|S|$ and that is Ramseyan for $\sigma$. We refer to this Ramseyan split as \emph{a Ramseyan split of $w$}.
\begin{theorem}\label{lemma:an-equivalence-class-contains-a-constructible-word}
	Let $\A=(Q, \Sigma, \Delta, I, F)$ be an automaton on scattered linear orderings. For any non-empty equivalence class $s$ induced by $\sim_\A$, the language $\varphi^{-1}(s)$ contains a constructible word.
\end{theorem}
\begin{proof}
	The main idea is to consider an arbitrary word $w\in \Sigma^{\textsf{scat}}$ of length $J$ and to show that $w$ is equivalent to a constructible word. This shows that the equivalence class of $w$ contains at least one constructible word.
	
	Given a Ramseyan split $r$ of $w$, we proceed by induction on the \emph{internal height} of $r$, defined as \[N^*=\left\{\begin{aligned}
		&\max\Bigl\{r(c)\mid c\in \widehat{J}^*\Bigr\} &\text{ if $\widehat{J}^*$ is non-empty}\\
		&\;0 &\text{ otherwise.}
	\end{aligned}\right.\]
	Since the height of $r$ is finite by Theorem~\ref{thm:colcombet}, so is its internal height. In what follows, we say that $w$ is of internal height $N^*$ if it admits a Ramseyan split of internal height $N^*$.
	
	Let $w$ be a word of internal height $0$. The set of non-extremal cuts $\widehat{J}^*$ of $w$ is empty, hence $w$ contains $0$ or $1$ letter. Hence in both cases, $w$ is constructible. In particular, $w$ is thus equivalent to a constructible word. 
	
	Now, let $w$ be a word of internal height $N^*>0$ and $r$ a Ramseyan split of $w$ with internal height $N^*$. Let $C\subseteq \widehat{J}$ be the sets of cuts of $w$ mapped to $N^*$ by $r$. Formally,
	\[C=\{c\in \widehat{J}\mid r(c)=N^*\}.\]
	Importantly, all the elements of $C$ are in the same $N^*$-class (indeed, they are all mapped to $N^*$ by $r$, and by maximality of $N^*$, any element of $\widehat{J}$ between two elements of $C$ is mapped to $k\leq N^*$).  Note that $C$ is necessarily non-empty, since it contains at least one non-extremal cut (otherwise, $r$ would not be a split of internal height $N^*>0$). We now distinguish two cases, depending on whether $C$ is a singleton or not. In each case, the idea of the proof is to use some cuts belonging to $C$ as separators to factorize $w$.
	\begin{itemize}
		\item If $C$ contains a single element $c$, consider the following factorization of $w$:
		\[w=w[\widehat{J}_\mathit{min}:c]\cdot w[c:\widehat{J}_\mathit{max}].\]
		The words $w[\widehat{J}_\mathit{min}:c]$ and $w[c:\widehat{J}_\mathit{max}]$ are of internal height at most $N^*-1$ (since none of the non-extremal cuts of $w[\widehat{J}_\mathit{min}:c]$ and $w[c:\widehat{J}_\mathit{max}]$ is mapped to $N^*$). Hence by induction hypothesis, they admit equivalent constructible words $u_1$ and $u_2$, respectively. Therefore, by Lemma~\ref{lemma:equivalence-relation}, the word $u_1\cdot u_2$ is equivalent to $w$, and it is a constructible word. 
		\item If $C$ contains strictly more than one element, then it contains two elements that are consecutive in $C$. Let $c<d$ be these two elements. Their existence relies on the fact that $J$ is scattered, hence $\widehat{J}$ is also scattered, and any subordering of $\widehat{J}$ as well. From the existence of $c$ and $d$, one can infer the existence of a word $w[c:d]$ of internal height at most $N^*-1$. By induction hypothesis, this word thus admits an equivalent constructible word, that we denote by $v$. It remains to show how to use the existence of $v$ to build a constructible word equivalent to $w$. We need to distinguish four cases, depending on the presence or the absence of a greatest/least element in $C$.
		\begin{itemize}
			\item Case 1: $C$ has a greatest and a least element, denoted respectively by $\cmax$ and $\cmin$. Then, consider the following factorization of $w$:
			\[w=w[\jmin:\cmin]\cdot w[\cmin:\cmax]\cdot w[\cmax:\jmax].\]
			Since $\cmin,\cmax,c$ and $d$ belong to the same $N^*$-class, $\sigma(\cmin,\cmax)=\sigma(c,d)$ by definition of a Ramseyan split, hence the word $w[\cmin:\cmax]$ is equivalent to $w[c:d]$, which is itself equivalent to the constructible word $v$. Moreover, the words $w[\jmin:\cmin]$ and $w[\cmax:\jmax]$ are of internal height at most $N^*-1$, since none of their non-extremal cuts is mapped to $N^*$ by the split $r$. Hence by induction hypothesis, there exist two constructible words $u_1$ and $u_2$ respectively equivalent to $w[\jmin:\cmin]$ and $w[\cmax:\jmax]$. Therefore, thanks to Lemma~\ref{lemma:equivalence-relation}, $w$ is equivalent to $u_1\cdot v\cdot u_2$, which is constructible.
			\item Case 2: $C$ has a least element $\cmin$ but does not have any greatest element. Let $\xi$ be the cofinality of $C$, and $c_\mathit{sup}\in \widehat{J}$ its least upper bound (which necessarily exists since $\widehat{J}$ is \emph{complete}, as stated in~\cite{BC07}). Consider a subordering $(c_i)_{i\in \xi}$ of $C$ that is cofinal in $C$, and such that $c_0=\cmin$. Now factorize $w$ as follows:
			\[w=w[\widehat{J}_\mathit{min}:\cmin]\cdot \prod_{i\in \xi}w[c_i:c_{i+1}]\cdot w[c_\mathit{sup}:\jmax].\]
			The situation is similar to the previous case.
			The words $w[\widehat{J}_\mathit{min}:\cmin]$ and $w[c_\mathit{sup}:\jmax]$ are of internal height at most $N^*-1$, hence by induction hypothesis they admit equivalent constructible words, say $u_1$ and $u_2$. Moreover, $\sigma(c_i,c_{i+1})=\sigma(c,d)$ for all $i\in \xi$ since $c,d,c_i$ and $c_{i+1}$ are in the same $N^*$-class. Subsequently, each subword $w[c_i:c_{i+1}]$ is equivalent to $w[c,d]$, which is itself equivalent to $v$. Therefore, by Lemma~\ref{lemma:equivalence-relation}, $w$ is equivalent to $u_1\cdot v^\xi\cdot u_2$, where $u_1$, $v$ and $u_2$ are all constructible.
			\bigbreak Note that $\xi$ is larger or equal to the ordinal $\omega$ (if it was not the case, $C$ would be empty or would have a greatest element). If $\xi=\omega$, then $u_1\cdot v^\xi\cdot u_2=u_1\cdot v^\omega\cdot u_2$ is constructible, since the $\omega$-power is a constructible operator. If $\xi>\omega$, then $\xi$ is an uncountable regular ordinal. Lemma~\ref{lemma:ordinal-powers-indistinguishable} then ensures that $u_1\cdot v^\xi\cdot u_2$ is equivalent to $u_1\cdot v^{\omega_1}\cdot u_2$, which is constructible. In both cases, $w$ is equivalent to a constructible word.
			
			\item Case 3: $C$ has a greatest element $\cmax$ but does not have any least element. This case is the mirrored version of the previous one. Instead of defining $\xi$ as the cofinality of $C$, it should be defined as the coinitiality of $C$. The least upper bound of $C$ needs to be replaced by its greatest lower bound $c_\mathit{inf}\in \widehat{J}$. The factorization of $w$ thus writes \[w=w[\jmin:c_\mathit{inf}]\cdot\left(\prod_{i\in -\xi}w[c_{i-1}:c_i]\right)\cdot w[\cmax:\jmax]\]
			where $(c_i)_{i\in -\xi}$ is a subordering of $C$ that is coinitial in $C$, and such that $c_0=\cmax$. With the same arguments as before, it is possible to prove that $w$ is equivalent to $u_1\cdot v^{-\omega}\cdot u_2$ or $u_1\cdot v^{-\omega_1}\cdot u_2$. These two words are constructible, since $u_1$, $u_2$ and $v$ are themselves constructible.
			\item Case 4: $C$ does not have a greatest element nor a least element. In this case, let $\xi_1$ be the coinitiality of $C$, and $\xi_2$ its cofinality. Let $c_\mathit{sup}$ and $c_\mathit{inf}$ be the least upper bound and the greatest upper bound of $C$. Let $(c_i)_{i\in -\xi_1}$ and $(d_i)_{i\in \xi_2}$ be two suborderings of 
			$C$ respectively coinitial and cofinal in $C$, and such that $c_0=d_0$. Factorize $w$ as follows:
			\[w=w[\jmin:c_\mathit{inf}]\cdot\left(\prod_{i\in -\xi_1}w[c_{i-1}:c_i]\right)\cdot \left(\prod_{i\in \xi_2}w[d_{i}:d_{i+1}]\right)\cdot w[c_\mathit{sup}:\jmax]\]
			Following the same reasoning as above, it can be proven that $w$ is equivalent to $u_1\cdot v^{-\omega}\cdot v^\omega \cdot u_2$, to $u_1\cdot v^{-\omega}\cdot v^{\omega_1}\cdot u_2$, to $u_1\cdot v^{-\omega_1}\cdot v^{\omega}\cdot u_2$ or to $u_1\cdot v^{-\omega_1}\cdot v^{\omega_1}\cdot u_2$. In each case, $w$ is equivalent to a constructible word.
		\end{itemize}
    \end{itemize}
    This concludes the proof.
\end{proof}

\begin{theorem}\label{thm:constructible-words-scattered}
    Let $\A_1$ and $\A_2$ be two automata on scattered linear orderings accepting the same constructible words. Then, $\A_1$ and $\A_2$ accept the same language.
\end{theorem}
\begin{proof}
    Let $\A$ be the automaton obtained by mere juxtaposition of $\A_1$ and $\A_2$ (i.e., $\A$ is the union of $\A_1$ and $\A_2$). We assume that the states of $\A_1$ and $\A_2$ are disjoint. 
   	If two words $u$ and $v$ are equivalent for $\A$, then these two words are also equivalent for $\A_1$ and for $\A_2$. Let $w$ be a word accepted by $\A_1$. Our goal is to show that $w$ is also accepted by $\A_2$. Thanks to Theorem~\ref{lemma:an-equivalence-class-contains-a-constructible-word}, there exists a constructible word $u$ such that $u\sim_\A w$. This implies that $u\sim_{\A_1}w$ and $u\sim_{\A_2}w$. Therefore, $u$ is also accepted by $\A_1$. Since $u$ is constructible, it is thus accepted by $\A_2$ as well. Finally, since $u\sim_{\A_2}w$, the word $w$ is also accepted by $\A_2$.
\end{proof}

\begin{remark}
	In the proof of Theorem~\ref{lemma:an-equivalence-class-contains-a-constructible-word}, we actually showed that any word $w$ admits an equivalent constructible word whose notation involves a bounded number of constructible operators. This comes from the fact that Colcombet's theorem guarantees the existence of a Ramseyan split of $w$ that is not only of finite height, but also of bounded height. Therefore, Theorem~\ref{thm:constructible-words-scattered} could actually be slightly strengthened. To show that two automata on scattered linear orderings accept the same language, it is not necessary to check that they accept the same constructible words. It suffices to show that it is the case for constructible words involving a bounded number of operators. We believe that this slight improvement might be important in future applications of this theorem.
\end{remark}
\section{Conclusion and Further Work}
We proved that rational languages of words indexed by scattered linear orderings are characterized by their constructible words. This generalizes the characterization of $\omega$-regular languages by their ultimately periodic words. The proof of our result is interesting in itself, since it does not rely on complementation but rather exploits semigroup theory.

We expect our result to be used in future developments. It is likely to simplify some proofs, since constructible words are simple objects to work with, in the sense that they admit a finite notation (which is not the case of non-constructible words). We also believe that it sheds light on the expressiveness of automata on scattered linear orderings.

It would be interesting to understand the connections between our result and other theorems of the same flavor in logic. For instance, a result of Läuchli and Leonard states that two formulas in the monadic first-order theory of order (MFO) are equivalent on scattered domains (i.e., have the same scattered models) if and only if they have the same simple (``constructible'') countable scattered models~\cite{Lauchli-Leonard} (see also~\cite[Corollary 7.10]{Ros82}). To the best of our knowledge, it is not known whether a language definable in MFO on unrestricted scattered domains is always rational. Hence it is not clear whether this result of Läuchli and Leonard aligns with Theorem~\ref{thm:constructible-words-scattered}, or is of a different nature. We would like to answer this question in some future work, by studying the languages definable in MFO.

Our result could potentially be extended to rational languages over unrestricted linear orderings. It would also be worthwhile to investigate whether it yields decidability results (for instance, regarding equivalence checking between two automata on scattered linear orderings).
\newpage
\bibliographystyle{plainurl} 
\bibliography{bib2doi}
\end{document}